\documentclass[11pt,twoside]{article}
\usepackage{acta-info}
\usepackage{euler}
\newcommand{\vol}{\textbf{4,} 1 (2012) 17--32}
\newcommand{\amss}{\amssubj{ 90C22, 15A45, 65F15}}
\newcommand{\CRs}{\CRsubj{G.2.2}}
\newcommand{\keyw}{\keywords{sum of squares, B\"ottcher-Wenzel
 biquadratic form, semidefinite programming, strict complementarity}}

\begin{document}

\title{Sum of squares representation for the
       B\"ottcher-Wenzel biquadratic form}
       \short{L. L\'aszl\'o}{Sum of squares representation}
\maketitle

\oneauthor
{\href{https://numanal.inf.elte.hu/~laszlo} {Lajos L\'ASZL\'O}}
{\href{https://numanal.inf.elte.hu}{Department of Numerical Analysis,}
\href{http://www.elte.hu/en}{E\"otv\"os Lor\'and University,} Hungary}
{\href{mailto:laszlo@numanal.inf.elte.hu}{laszlo@numanal.inf.elte.hu}}

\begin{abstract}
We find the minimum scale factor, for which the nonnegative
B\"ottcher-Wenzel biquadratic form becomes a sum of squares (sos).
To this we give the primal and dual solutions for the underlying
semi-definite program.  Moreover, for special matrix classes
(tridiagonal, backward tridiagonal and cyclic Hankel matrices)
we show that the above form is sos.
Finally, we  conjecture sos representability for Toeplitz matrices.
\end{abstract}

\section{Introduction}
The B\"ottcher-Wenzel inequality (\cite{BW1}, \cite{BW2}, \cite{Aud}, \cite{VJ}, \cite{LU})
states that for real square matrices $P, Q$ the biquadratic form
\begin{equation} \label{bwform} BW \equiv 2 \, \Big(||P||^2||Q||^2-
\mathrm{trace}^2(P^{T}Q)\Big) - ||PQ-QP||^2 \end{equation}
is nonnegative, with the Frobenius norm used.
Replacing the factor $2$ by $2+\gamma_n,$ it is natural to ask
for the minimum $\gamma_n$ such that
\[ (2+\gamma_n) \, \Big(||P||^2||Q||^2-
\mathrm{trace}^2(P^{T}Q)\Big) - ||PQ-QP||^2  \]
is a sum of (polynomial) squares. We answer this question in
 Theorem \ref{mainthm} by showing that the minimum value is
$(n-2)/2.$

For simplicity, we use one-subscript notation for the entries of
$P$ and $Q,$  described by means of the "small" index matrix
\[\mathrm{IND} = ((i-1)n+j)_{i,j=1}^n.\]
Then $P$ and $Q$ can be generated by vectors $p$ and $q$
of dimension $m=n^2$ as
\[ P(i,j)=  p\,(\mathrm{IND}(i,j)), \ 1\le i\le j\le n. \]
Introducing an index matrix will be  especially useful in
Sections \ref{sec3} to \ref{sec5}, where  tridiagonal, backward tridiagonal,
cyclic Hankel and general Toeplitz matrices
will be investigated.
For these special cases we prove (for Toeplitz matrices:
conjecture) that the corresponding $BW$ form is a sum of squares (sos).

It is quite odd that although (real) Hankel matrices are symmetric,
thus normal, hence nonnegativity easily follows \cite{BW1},
this does not imply that a sos form also exists
(except if $n=3$, Example \ref{hankel3}).
On the other hand, Toeplitz matrices are usually not normal, yet
the corresponding BW form is sos, at least according to
our well-grounded Conjecture \ref{conj} at the end of the paper.

\section{The case of general matrices}
Let $P, Q$ be arbitrary $n\times n$ real matrices with entries
 \[P=(p_{(i-1)n+j})_{i,j}^n, \quad Q=(q_{(i-1)n+j})_{i,j}^n,\]
as indicated above.
(Notice that we use this indexing technique for simplicity.)
It turns out \cite{LL3} that
the above forms depend only on the variables
\[ z_{i,j} = p_iq_j-q_ip_j, \ 1\le i<j\le n,\]
i.e. on the skew symmetric matrix
\begin{equation*}  Z=pq^T-qp^T \end{equation*}
 of order $n^2,$ a benefit of including the term
 trace$^2(P^{T}Q)$.
Indeed, we have
\[ ||Z||^2 = 2 \, \Big( ||P||^2||Q||^2-
\mathrm{trace}^2(P^{T} Q) \Big),\]
 and all entries of the commutator
$[P, Q]$ obviously are linear forms of the $z_{i,j}$-s.
\medskip

Let us formulate the primal and dual semi-definite programming
problems (see e.g. in \cite{Todd}) for the  eigenvalue optimization:
\begin{eqnarray*} &&\mathrm{min}\ \{ \mathrm{tr}(CX): \ X\ge 0, \
\mathrm{tr} (A_iX) = 0,\ 1\le i\le M, \
 \mathrm{tr}(X) = 1\} \quad  (Primal) \cr
  &&\mathrm{max}\ \{y_{M+1}: \ S \equiv C -
\sum_{t=1}^M y_t A_t - y_{M+1} I \ge 0\} \ \ \quad \quad \
\qquad \qquad (Dual) \end{eqnarray*}
where $C,\, S,\, X,\, A_t$ and the identity $I=I_N$ are all real
symmetric $N$th order matrices,
$C$ and $(A_t)_1^M$ are given, the primal matrix $X,$
the dual (slack) matrix $S$ and vector $y$ are the
solutions of the program, tr$(AB)\equiv$ trace$(AB)$ denotes the scalar
product of the symmetric matrices $A$ and $B,$ and $\ge$ stands
for the semi-definite ordering: \ $A\ge B$ iff $A-B$
is positive semi-definite.
\bigskip

The quantities $z_{i,j}$ will play the role
of 'candidate monomials' (better to say, differences,
and hereafter called candidates) with ordering
\[ z = ( z_{1,2}, \ z_{1,3},\ z_{2,3},\ z_{1,4},\ \dots,
 \ z_{1,n}, \ \dots, z_{n-1,n} )^T.\]
The indices can be read from the "big" index matrix
\[ \mathrm{POS} =
\left(\begin{array} {rrrrrr} 0 & 1 & 2 & 4 & 7 & \dots \cr
\cdot & 0 & 3 & 5 & 8 & \dots \cr
\cdot & \cdot & 0 & 6 & 9 & \dots \cr
\cdot & \cdot & \cdot & 0 & 10 & \dots \cr
\vdots & \vdots & \vdots & \vdots & \vdots & \ddots \cr
\end{array}\right) \]
\vspace*{-6pt}
of order $n^2$ to be
\[(i,j) \sim k\equiv i+\frac{(j-1)(j-2)}{2}, \quad 1\le i<j\le n^2. \]
Note that $z_{i,i}=0$ for all $i,$ and that the entries
below the diagonal are omitted due to $z_{i,j}=-z_{j,i},$
enabling us to reduce the number of unknowns.
Also note that IND is related to $P$
and $Q,$ while POS is connected with $Z.$

As an example, we give the biquadratic form $BW$ as a
quadratic form of the quantities $(z_{i,j})$ for $n=3.$ Observe that
$ \|Z\|^2 = \|Z\|^2_F= 2\sum_{1\le i<j\le 9} z_{i,j}^2.$ \medskip

{\it Example 1.} For $n=3$ the objective takes the form
\begin{eqnarray*} BW &=& \|Z\|^2 
   -(z_{2,4}+z_{3,7})^2 -(z_{1,2}+z_{2,5}+z_{3,8})^2
 -(z_{1,3}+z_{2,6}+z_{3,9})^2 \cr
  & -&(-z_{1,4}-z_{4,5}+z_{6,7})^2   -(-z_{2,4}+z_{6,8})^2
  -(-z_{3,4}+z_{5,6}+z_{6,9})^2 \cr
  &  -&(-z_{1,7}-z_{4,8}-z_{7,9})^2 -(-z_{2,7}-z_{5,8}-z_{8,9})^2
  -(-z_{3,7}+z_{6,8})^2.  \end{eqnarray*}
(Note that the form $BW$ can be thought of as a function of the
matrices $(P,Q),$ the vectors $(p,q),$ the matrix $Z,$ or of
the vector $z.$)  We give now all the (quadratic) relations
holding for the variables $(z_{i,j})_{1\le i<j\le n^2}$ as
\begin{equation} \label{basic} z_{i,j}\, z_{k,l} + z_{i,l}\, z_{j,k} -
z_{i,k}\, z_{j,l} = 0, \quad 1\le i<j<k<l\le n^2. \end{equation}
These easily checked relations define $M=\binom{n^2}{4\ }$ symmetric
constraint matrices $A_t,$ each having
exactly 6 nonzero (off-diagonal) entries. For instance, equation
\begin{equation*} z_{2,3}\, z_{4,5} +
 z_{2,5}\, z_{3,4} - z_{2,4}\, z_{3,5} = 0 \end{equation*}
defines an $A_t$ with nonzero entries in
positions $(3,10), \, (8,6), \, (5,9)$ and their transposes,
see the matrix POS.  Now we can state our main theorem.

\begin{theorem} \label{mainthm} The minimum value of $\gamma_n,$
for which (\ref{bwform}) is a sum of squares is
\[\gamma_n=\frac{n-2}{2}.\]  \end{theorem}
\begin{proof}
We give the optimal primal and dual solutions
and describe the main characteristics of the optimal dual matrix.
Since the objectives coincide, the strong duality theorem
yields the desired result.

By fixing the order of the variables $(z_{i,j})$ above,
matrix $C$ is uniquely determined. To get the (slack) matrix
$S=C-\sum y_t A_t,$ we use the following strategy.
Note that we not only give the set $(A_t)$ of active constraints
(as e.g. when taking the half Newton-polytope), but also give
their coefficients $(y_t).$ \smallskip

{\it Strategy A.} Assume the commutator $[P,Q]$ contains an entry
 $(z_{i,j}+z_{k,l}+\dots )$ with $i,j,k,l$ distinct
and $i<j, \ k<l.$
Then the quadratic form $z^TCz$ associated with $C$ necessarily
contains a term $2\,z_{i,j}\,z_{k,l}.$ We 'halve' this
term, and leave one $z_{i,j}\,z_{k,l}$ unchanged as is,
while apply for the other term the basic quadratic relation
(\ref{basic}).
By using the correct sign, this defines a constraint $A_t$
and the corresponding dual variable $y_t$ for some $t.$
Finally, let $y_{m+1}=-\frac{n-2}{2}.$  \smallskip

Now we give the obtained  primal and dual solutions.
In view of the quite combinatorial character of the problem,
we do not detail each block, instead we give
some explanations and important cross-references
(control sums) and for matrix $S$ we provide Table \ref{tabgen}.
with all essential informations. \medskip

{\sl The Primal Problem} \bigskip

Before defining the optimal primal matrix $X,$
we note that its rank is $\binom{n}{2}.$
For indices $(i,j): 1 \le i < j \le n$ we define the
vectors $v_{i,j}$ of dimension $\binom{n^2}{2}$ to have
$4(n-1)$ nonzero coordinates (four $2$'s and $4(n-2)$
$\pm 1$'s) using row$_i,$ row$_j,$ column$_i$ and
column$_j$ of the index matrix IND, cf. Example
\ref{rowcol} below.

Next we form the matrix of these vectors
\[V=[v_{1,2},\ v_{1,3,}\ v_{2,3},\ \dots,\ v_{n-1,n}], \]
and define matrix $X_0=V\, V^T =\sum v_{i,j}v_{i,j}^T$
with the following properties:

$X_0$ is a symmetric matrix of order $N=\binom{n^2}{2}$ and rank
 $\binom{n}{2}.$ The $v_{i,j}'$'s are orthogonal with norm square
  $\|v_{i,j}\|^2=4\cdot 4+4(n-2)\cdot 1 = 4(n+2).$
The trace of $X_0$ is
\[\mathrm{tr}(X_0)=\sum_{i<j}\mathrm{tr} (v_{i,j}v_{i,j}^T)
=\sum_{i<j}\|v_{i,j}\|^2=4(n+2)\binom{n}{2}, \]
thus by defining
\[ X=\Big(4(n+2)\binom{n}{2}\Big)^{-1} X_0\]
we get a trace 1 matrix. The eigenvalues of $X_0$ are $4(n+2),$
those of $X$ are $\binom{n}{2}^{-1}$
(hence $\binom{n}{2}^{-1}X$ is a projection).
The $v_{i,j}$'s are also eigenvectors of $C$:
\[ C \, v_{i,j}=(2-n)/2 \ v_{i,j}.\]

Furthermore we have
\[\mathrm{tr}(CX_0)=\sum_{i<j}\mathrm{tr}(Cv_{i,j}v_{i,j}^T)
=\sum_{i<j} v_{i,j}^TCv_{i,j}=
\frac{2-n}{2}\sum_{i<j}\|v_{ij}\|^2, \]
and finally, the primal objective equals
\begin{equation} \label{prim}\mathrm{tr}(CX) =
\frac{\mathrm{tr}(CX_0)}{\mathrm{tr}(X_0)}
= \frac{2-n}{2}. \end{equation}

{\sl The Dual Problem} \medskip

The matrix $S=C-\sum y_t A_t$ resulting from Strategy A
is positive semi-definite and decomposes into some blocks
given in Table \ref{tabgen}.
(Observe that all eigenvalues and diagonal entries
of $2 S$ are integer -- a reason for the factor $2.$)

\begin{table} \begin{center}
\begin{small}
\begin{tabular}{||c|c|c|c|c|c|c||} \hline \hline
1 & No of blocks & $\binom{n}{2}$ & 1 &  $ 3\binom{n}{4} $ &
   $\binom{n}{2}$  & Total \\ \hline
2 & Block sizes & $ 6n-8$ & $ \binom{n}{2}$ & 4 & 1 & \\ \hline \hline
3 &  Eig=0 & 2 & n-1 & -- & -- & $n^2-1$ \cr \hline
4 &  Eig=4 &  1 & -- & -- & -- & $n(n-1)/2$ \cr \hline
5 &  Eig=$n$ & $2n-4$ & $\binom{n-1}{2}$ & 1 & -- &
         $(n^2-1)(n-2)(n+4)/8$ \cr \hline
6 & Eig=$n+2$& $3n-5$ & -- & 2 & 2 & $(n^2-1)(n-2)(n+4)/4$  \cr \hline
7 &  Eig=$n+4$ & $n-3$ & -- & 1 & -- & $n(n-1)(n^2+n-2)/8$  \cr \hline
8 & Eig=$2n+2$ & 1 & --  & -- & -- & $n(n-2)/2$  \cr \hline \hline
9 &  Diag & $n(+2)$ & $n-2$  & $n+2$ & $n+2$ & $(n^3-n)(n^2+2n-4)/2$
 \cr \hline \hline
\end{tabular}\end{small} \end{center}
\caption{Decomposition of the matrix "2S" \label{tabgen}}
\end{table}

Here we list the important facts and control sums
concerning the blocks of $S,$ and in Example (\ref{ExDual}) we give
further hints for understanding the construction. \smallskip

ROW-control: an element in the last column of row $i$
is the scalar product of row 1 and row $i$.
For instance, the number of zero eigenvalues---the defect
of $S$---equals to $\binom{n}{2} 2 + 1 (n-1) = n^2-1.$

EIG-control: the last column (the number of eigenvalues, rows 3 to 8)
sums up to $\binom{n^2}{2},$ the order of the matrix $2S.$

DIAG-control: The sum of the elementwise products of row 1, 2 and 9,
\[ \binom{n}{2} \Big( 2n*n + 4(n-2)(n+2)\Big)
+1\binom{n}{2}(n-2)+3\binom{n}{4}4(n+2)+\binom{n}{2}1(n+2)\]
equals to $\binom{n}{2}(n+1)(n^2+2n-4),$ the trace of $2S.$
(In the blocks of order $6n-8$ there are $2n$
diagonal elements "$n$", and $4(n-2)$ diagonal elements "$(n+2)$".)

TRACE-control: the trace of the coefficient matrix $C$ equals
\[\mathrm{tr}(C)=2\binom{n^2}{2}-n(n-1)-(n^2-n)n = n(n-1)^2(n+1).\]
The first subtrahend comes from the  diagonal of the
commutator $[P,Q],$ the second from their off-diagonal elements.
Due to $\mathrm{diag}(S)=\mathrm{diag}(C)+\gamma_n I,$
the connection between the traces of matrices $C$ and $S$ is
\[\mathrm{tr}(2S) = 2 \Big(\mathrm{tr}(C) +
\frac{n-2}{2}\binom{n^2}{2}\Big).\]
The number of all constraints is $\binom{n^2}{4},$
while that of active constraints equals
\[ n\binom{n-1}{2}+n(n-1) \Big( \binom{n-2}{2}+2(n-2) \Big)
= \binom{n}{2}(n^2-4) \equiv 3(n+2)\binom{n}{3}. \]
Here the first term is associated with the main diagonal
of $R\equiv[P,Q]$ (by virtue of $z_{i,i}=0$ there are only $n-1$
terms in $R(i,i)$), while the rest comes from the off-diagonal
of the commutator $R$ (where there always are two terms for
which the basic relations do not apply, see the Example \ref{rowcol}).

Thus $S$ is positive semi-definite with defect $n^2-1,$ and
its eigenvalues range in the interval $[0, n+1].$
To sum up, the primal objective (\ref{prim})
coincides with the dual objective $y_{M+1},$ the negative
of $\gamma_n,$ which proves the theorem.  \end{proof}

There holds no strict complementarity, for
$\mathrm{rank}(X) = \binom{n}{2} < n^2-1 = \mathrm{def}(S).$

\begin{example} \label{rowcol}
To define the primal matrix $X$ take the four scalar products
\[\langle\mathrm{row}_i,\mathrm{col}_j\rangle, \quad
  \langle\mathrm{col}_i,\mathrm{row}_j\rangle, \quad
  \langle\mathrm{row}_i,\mathrm{row}_j\rangle, \quad
  \langle\mathrm{col}_i,\mathrm{col}_j\rangle\]
in the index matrix IND where
each of the four products determine $n$ coordinates in $v_{i,j}$
as follows. If $n=3$ and $i=1, \ j=2,$ then
$\mathrm{row}_1=[1,2,3], \ \mathrm{col}_2=[2,5,8]^T$ which yields by
$(1,2)\sim 1, \ (2,5)\sim 8, \ (3,8)\sim 24$
the coordinates $1, \, 8, \, 24,$ see also matrix POS.
Similarly we calculate the other three triples, giving together
\[ 1,\, 8,\, 24; \ 4,\, 10,\, 21\, (!) ; \ 4,\, 8,\, 13;
 \ 1,\, 10,\, 28. \]
The repeated elements (1, 4, 8, 10) denote positions with value $2.$
The exclamation sign refers to an entry $-1$ (since $(7,6)$ must be
inverted to $(6,7)\sim21$).  To sum up, we get
\begin{eqnarray*} v_{1,2}=&&
 (2, 0, \ 0, 2, 0, 0, 0, 2, 0, 2, 0, 0, 1, 0, 0, 0, 0, 0,\cr
&& 0, 0, -1, 0, 0, 1, 0, 0, 0, 1, 0, 0, 0, 0, 0, 0, 0, 0)^T.
\end{eqnarray*} \end{example}

\begin{example} \label{ExDual} Hints for obtaining the dual matrix.
We give some details for the $\binom{n}{2}$ most important
blocks of order $6n-8.$
There is a one-to-one correspondence between these blocks
and ordered pairs $(i,j), \ 1\le i<j\le n.$
To collect the indices for the block containing
$z_{i,j},$ we have to consider the $4(n-1)$ terms in
\[\langle \mathrm{row}_i, \mathrm{col}_j \rangle, \quad
  \langle \mathrm{col}_i, \mathrm{row}_j \rangle, \quad
  \langle \mathrm{row}_i, \mathrm{row}_j \rangle, \quad
  \langle \mathrm{col}_i, \mathrm{col}_j \rangle \]
  (the same as for $v_{i,j}$ above!)
and further $2(n-2)$ terms in the products
\[ \mathrm{IND}(i,j)* \mathrm{diag}(\ne i,j), \quad
\mathrm{IND}(j,i)*\mathrm{diag}(\ne i,j). \]
Here $\mathrm{diag}(\ne i,j)$ stands for the $n-2$ entries of the
diagonal of IND, differing from $i, j.$
As in  \emph{Example \ref{rowcol}}, choosing $n=3,\ i=1,\ j=2,$
vector diag$(\ne 1,2)$ reduces to the $(3,3)$ entry 9, thus
we get (using index matrices IND and POS)
\[((1,2),(3,3))\sim (2,9) \sim 30 \quad \mathrm{and} \quad
((2,1),(3,3))\sim (4,9) \sim 32. \]
Hence the diagonal block containing row 1 (related to $z_{1,2}$)
also contains rows 30 and 32. The whole index set at issue is \
$ [1, \ 4, \ 8, \ 10, \ 13, \ 21, \ 24, \ 28, \ 30, \ 32 ],$
and the corresponding block is the $10\times 10$ (irreducible) matrix
\[\left( \begin{array}{rrrrrrrrrr}
 3 & 0 &-2 & 0 &-1 & 0 &-1 & 0 & 0 & 0 \cr
 0 & 3 & 0 &-2 & 0 & 1 & 0 &-1 & 0 & 0 \cr
-2&  0&  3 & 0&  0&  0& -1& -1&  0&  0 \cr
 0& -2&  0 & 3& -1&  1&  0&  0&  0&  0 \cr
-1&  0&  0& -1&  5&  0&  0& -1& -1&  1 \cr
 0&  1&  0&  1&  0&  3& -1&  0&  0&  0 \cr
-1&  0& -1&  0&  0& -1&  3&  0&  0&  0 \cr
 0& -1& -1&  0& -1&  0&  0&  5&  1& -1 \cr
 0&  0&  0&  0& -1&  0&  0&  1&  5&  0 \cr
 0&  0&  0&  0&  1&  0&  0& -1&  0&  5
\end{array} \right) \]
with eigenvalues $(0,\ 0,\ 3,\ 3,\ 4,\ 5,\ 5,\ 5,\ 5,\ 8).$
\end{example}

\section{Tridiagonal (and backward tridiagonal) matrices}\label{sec3}
In a former paper \cite{LL2} we have shown that for $n$th order
matrices $P, Q$ with only nonzero entries in row $1$ and column $n$
the BW form is sos, however in case of (additional) main diagonal
elements this is no more true. Therefore one would guess that
$3n+O(1)$ nonzero elements cannot be allowed, however the result
below shows that the answer depends on the position of these elements.

We shall use an index matrix given e.g. for
$n=3$ as \ IND =$\Big(\begin{smallmatrix}
1 & 2 & 0 \cr 3 & 4 & 5 \cr 0 & 6 & 7 \cr \end{smallmatrix} \Big).$

\begin{lemma} For tridiagonal $P, Q$ the BW form is sos, especially we have
\begin{eqnarray*}
 BW&=&2\sum_{i<j}z_{i,j}^2 \cr &-& \sum (z_{3i-4,3i-3}-z_{3i-1,3i})^2
  -\sum z_{3i-2,3i-1}^2 - \sum z_{3i,3i+3}^2 \cr
  &-& \sum (z_{3i-2,3i-1}+z_{3i-1,3i+1})^2-\sum (z_{3i-2,3i}+z_{3i,3i+1})^2 \cr
  &=& \sum (z_{3i-4,3i-1}+z_{3i-3,3i})^2 + (z_{3i-4,3i}-z_{3i-3,3i-1})^2 \cr
  &+& \sum (z_{3i-2,3i-1}-z_{3i-1,3i+1})^2 + \sum (z_{3i-2,3i}-z_{3i,3i+1})^2 \cr
  &+& \sum z_{3i,3i+2}^2 + 2 \sum z_{3i-2,3i+1}^2
   + 2 \sum_{i+5\le j} z_{i,j}^2 - \sum z_{3i-1,3i+3}^2. \end{eqnarray*}
\end{lemma}

\begin{remark}
The first equality gives the biquadratic form at issue,
the second one is the claim: the sum of squares representation.
(The negative terms in the last row are evidently canceled.)
\end{remark}
Although SDP is not needed here, for the identity of the Lemma
can be proved directly, we yet give some facts.
The eigenvalues of the dual matrix $S$ for the actual semidefinite
programming problem are integers ($0, 1, 2, 3$) in this case, too.
This is so because matrix $S$ decomposes into at most
second order blocks of the form
$(\begin{smallmatrix} 1 & 1 \cr 1 & 1 \end{smallmatrix})$ and
$(\begin{smallmatrix} 2 & -1 \cr -1 & 2 \end{smallmatrix}).$

Table \ref{tabtri} illustrates the main features of the underlying
semidefinite program.
First the number of the eigenvalues of $S$ are given,
then the number of $2\times 2$ blocks in $S$
(the number of scalar blocks is not shown),
the number of the active ($y_t\ne 0$) constraints,
and finally, the rank of $X.$ (The $n-2$ active constraints
correspond to the positions $(i-1,i), (i,i-1),
(i,i+1)$ and $(i+1,i)$ in IND.) \smallskip

It is easy to get a formula for these quantities, e.g.
the number of eigenvalues $\lambda_i=2$
can be determined by subtracting the number of all other
eigenvalues from the order $ (3n-2)(3n-3)/2 $ of $S.$
The result is $3+9\binom{n-1}{2}.$

Note that strict complementarity does hold:
the number of zero eigenvalues of $S$ coincides with
rank$(X),$ the number of nonzero eigenvalues of $X.$

\begin{table} \begin{center}
\begin{tabular}{||c|c|c|c|c|c|c|c||} \hline \hline
 $n$ & $\lambda=0$ & $\lambda=1$ &  $ \lambda=2 $ &
   $\lambda=3$  &  2-bl. & act. & rk(X) \\ \hline \hline
2 & 3 & 0 & 3 & 0 & 2 & 0 & 3  \\ \hline
3 &  7 & 1 & 12 & 1 &  6 & 1 & 7 \cr \hline
4 &  11 & 2 & 30 & 2 &  10 & 2 & 11 \cr \hline
5 &  15 & 3 & 57 & 3 & 14 & 3 & 15 \cr \hline
6 & 19 & 4 & 93 & 4 & 18 & 4 & 19 \cr \hline
7 & 23 & 5 & 138 & 5 & 22 & 5 & 23  \cr \hline
8 & 27 & 6 & 192  & 6 & 26 & 6 & 27 \cr \hline \hline
\end{tabular} \end{center} \caption{``Tridiagonal matrices''
\label{tabtri}}\end{table}\medskip

{\bf Backward tridiagonal matrices} \medskip

They have many similar properties,
except that the case $n$ odd is worse:
while for $n$ even all the eigenvalues
of $S$ are integers (lying in  $[0, 4]$), for $n$ odd
this does not hold, therefore we write '--' instead.
Also, in this case there are (apart from the
scalar and $2\times 2$ blocks) $4\times 4$ blocks, too.
All this information is contained in Table \ref{taback}
from where one can see that for $n$ even we again have
strict complementarity, as in the tridiagonal case.

\begin{table} \begin{center}
\begin{tabular}{||c|c|c|c|c|c|c|c|c|c||} \hline \hline
 $n$ & $\lambda=0$ & $\lambda=1$ &  $ \lambda=2 $ &
   $\lambda=3$  & $\lambda=4 $ & act & 2-bl. & 4-bl.
   & rk(X) \\ \hline \hline
2 & 3 & 0 & 3 & 0 & 0 &  2 & 0 & 0 & 3\\ \hline
3 &  8 & - & - & - & - & 3 & 3 & 7 & 5\cr \hline
4 & 13 & 4 & 25 & 0 & 3 & 6 & 3 & 8 & 13\cr \hline
5 & 20 & - & - & - & - &  11 & 6 & 17 & 18 \cr \hline
6 & 25 & 6 & 81 & 2 & 6 & 14 & 6 & 18 & 25 \cr \hline
7 & 32 & - & - & - & - & 30 &  9 & 27 & 30 \cr \hline
8 & 37 & 8 & 173  & 4 & 9 & 22 & 9 & 28 & 37 \cr \hline \hline
\end{tabular} \end{center} \caption{"Backward tridiagonal matrices"
\label{taback}}\end{table}

\begin{example} We calculate the number {\sl act} of active constraints:
\[ \mathrm{act} = \begin{cases} 5n-12, \  n \ \textrm{even} \\
   5n-8, \ \ n \ \textrm{odd} \end{cases}\]
The number of terms in a typical row of the product of
backward tridiagonal matrices usually equals
$(0,\dots, 0, 1, 2, 3, 2, 1, 0,\dots, 0).$
 However, in case of the commutator $PQ-QP$ there are
 some minor changes:
for odd order 1, for even order 2 main diagonal entries contain
only two terms (instead of 3), due to the identity $z_{i,j}+z_{j,i}=0.$
On the other hand, if $n$ is even, there are two opposite
entries (with indices $(k,k+1)$ and $(k+1,k),$ where $k=n/2$)
which do not generate any constraint, for the corresponding
indices are {\sl not} distinct.

Now we easily calculate the number asked, which is e.g.
for $n=6$ equal to $5n-12=18.$ To this consider the matrix
\[\begin{pmatrix}
2&2&1&0&0&0\cr2&3&2&1&0&0\cr1&2&2&*&1&0\cr
0&1&*&2&2&1\cr0&0&1&2&3&2\cr0&0&0&1&2&2
    \end{pmatrix}\]
with the number of terms in a given position of $[P,Q],$
and take $\binom{e}{2}$ for any entry $e>1.$
They sum up to $2*\Big(6\binom{2}{2}+\binom{3}{2}\Big)=18.$
The general case is similar. \end{example}

\section{Cyclic Hankel matrices}
When investigating Hankel matrices, we find that -- except for the
case $n=3,$ see below -- they do not generate sos BW forms.
However, cyclic ones behave well. We make use of the small index
matrix (given for $n=3$): \ IND =$\Big(\begin{smallmatrix}
1 & 2 & 3 \cr 2 & 3 & 1 \cr 3 & 1 & 2 \cr \end{smallmatrix} \Big).$

\begin{theorem}
For cyclic Hankel matrices $P, Q$ the BW form is
a sum of squares. \end{theorem}

\begin{proof}
Using the above-defined index matrix with $(1, 2,\dots, n)$ as
first row and $(n, 1,\dots, n-1)^T$ as last column, we obtain
\[ \|P\|^2=n\|p\|^2, \ \ \|Q\|^2=n\|q\|^2, \ \
\mathrm{trace}(P^TQ)=n p^Tq, \]
consequently
\[ 2 \, \Big(||P||^2||Q||^2-\mathrm{trace}^2(P^{T}Q)\Big) =
2n^2 (\|p\|^2 \|q\|^2-(p^Tq)^2).\]
The commutator $[P,Q]$ is a skew symmetric cyclic Toeplitz
matrix having
\[k=k_n=\Big[\frac{n-1}{2}\Big]\]
different entries $t_i=t_i^{(n)}$ with row one as
\begin{eqnarray*}
&&( 0, t_1, \dots, t_k, \, -t_k, \dots, -t_1) \quad \
\mathrm{(n\ \, odd\,)} \cr &&( 0, t_1, \dots, t_k, 0, -t_k,
\dots, -t_1) \ \ \mathrm{(n\ even).}
\end{eqnarray*}
Thus the subtrahend is $\|R\|^2=2 n \sum t_i^2,$ and the whole BW form equals
\[2 n\, \Big( n \sum_{i<j}^n z_{i,j}^2 - \sum_1^k t_i^2\Big). \]
Observe now that all terms in
\[ t_i=t_i^{(n)}= \sum_{j=1}^{n-i} z_{j,i+j}-\sum_{j=1}^i z_{j,n-i+j}\]
are distinct $(i=1,\dots,k),$ hence the Cauchy-Schwarz inequality
in conjunction with
$n k = n \Big[\frac{n-1}{2}\Big] \le \binom{n}{2}$ imply
\[ \sum_{i=1}^k t_i^2 \le \sum_{i=1}^k n
\Big(\sum_{j=1}^{n-i}  z_{j,i+j}^2 + \sum_{j=1}^i z_{j,n-i+j}^2  \Big)
\le n \sum_{i<j}z_{i,j}^2, \] which proves the theorem.
The last inequality turns into equality for $n$ odd.
\end{proof}

\begin{remark}  The case $n=4$ is especially interesting.
Then the commutator is
\[ \left(\begin{array} {rrrr} 0 & t & 0 & -t \cr -t & 0 & t & 0 \cr
0 & -t & 0 & t \cr t & 0 & -t & 0 \end{array} \right) \] with
$t=t_1=t_1^{(4)}=z_{1,2}+ z_{2,3}+ z_{3,4}- z_{1,4},$ therefore the
formula
\begin{eqnarray*} & & 4\,
( z_{1,2}^2+ z_{1,3}^2+ z_{2,3}^2+ z_{1,4}^2+
z_{2,4}^2+ z_{3,4}^2 )= \cr
&+& ( z_{1,2}+ z_{2,3}- z_{1,4}+ z_{3,4} )^2
+ ( z_{1,2}- z_{2,3}+ z_{1,4} +z_{3,4} )^2 \cr
&+& ( z_{1,2}+ z_{1,3}- z_{2,4}- z_{3,4} )^2
+ ( z_{1,2}- z_{1,3}+ z_{2,4}- z_{3,4} )^2 \cr
&+& ( z_{1,3}+ z_{2,3}+ z_{1,4}+ z_{2,4} )^2
+ ( z_{1,3}- z_{2,3}- z_{1,4}+ z_{2,4} )^2, \end{eqnarray*}
(a consequence of Eulers identity)
yields the sos-representation needed. \end{remark}

\begin{example} \label{hankel3}
The case of (general) third order Hankel matrices.
The index matrix IND is now
$\Big(\begin{smallmatrix} 1 & 2 & 3 \cr 2 & 3 & 4 \cr 3 & 4 & 5
\end{smallmatrix}\Big),$ the order of $C, \ S$ and
the constraint matrices $\{A_t\}$ is $\binom{5}{2}=10,$
the number of the $A_t$-s is $\binom{5}{4}=5.$
By help of  vector
\[ z = ( z_{1,2}, \ z_{1,3},\ z_{2,3},\ z_{1,4},\ z_{2,4},
    \ z_{3,4}, \ z_{1,5}, \ z_{2,5}, \ z_{3,5}, \ z_{4,5})^T\]
and matrix $C$
the objective can be written as $BW=z^TCz=\mathrm{tr}(Czz^T),$
which becomes -- by means of a standard SDP relaxation --
trace$(CX).$  Our MATLAB program yields $y=(0,0,1,0,0,0),$
i.e. only one constraint  will be  active, giving
\[S = C-y_3A_3=\left( \begin{array} {rrrrrrrrrr}
 1 & 0 &-1 & 0 & 0 &-1 & 0 & 0 & 0 & \{1\} \cr
 0 & 2 & 0 & 0 &-1 & 0 & 0 & 0 &-1 & 0 \cr
-1 & 0 & 4 & 0 & 0 &-2 & 0 & 0 & 0 &-1 \cr
 0 & 0 & 0 & 2 & 0 & 0 & 0 & \{-1\} & 0 & 0 \cr
 0 &-1 & 0 & 0 & 3 & 0 & \{1\} & 0 &-1 & 0 \cr
-1 & 0 &-2 & 0 & 0 & 4 & 0 & 0 & 0 &-1 \cr
 0 & 0 & 0 & 0 & \{1\} & 0 & 1 & 0 & 0 & 0 \cr
 0 & 0 & 0 & \{-1\} & 0 & 0 & 0 & 2 & 0 & 0 \cr
 0 &-1 & 0 & 0 &-1 & 0 & 0 & 0 & 2 & 0 \cr
 \{1\} & 0 &-1 & 0 & 0 &-1 & 0 & 0 & 0 & 1 \cr
\end{array}\right). \]
(In the original $C$ the six entries in braces are zero.)
The last zero in $y$ indicates the sos representability.
To obtain the concrete sos form, we calculated the
eigen-decomposition of the three blocks
\[B_1=\left(\begin{array}{rrrr} 1 & -1 & -1 & 1 \cr -1 & 4 & -2 & -1 \cr
-1 & -2 & 4 & -1 \cr 1 & -1 & -1 & 1 \end{array} \right),
B_2=\left(\begin{array} {rrrr}2 & -1 & 0 & -1 \cr -1 &
3 & 1 & -1 \cr 0 & 1 & 1 & 0 \cr -1 & -1 & 0 & 2
\end{array}\right),
B_3=\left(\begin{array} {rr}2 & -1 \cr -1 & 2 \end{array}\right)\]
with integer eigenvalues
\[ E_1: \ \ \begin{pmatrix} 0 \ & 0 \ & 4 \ & 6 \ \end{pmatrix}, \qquad
   E_2: \ \begin{pmatrix} 0 \ & 1 \ & 3 \ & 4 \ \end{pmatrix}, \qquad
   E_3: \ \begin{pmatrix} 1 & 3 \end{pmatrix} \]
and (unnormalized, integer, columnwise) eigenvectors
\[ V_1:\ \left(\begin{array}{rrrr}
2 & -1 & 1 & 0 \cr 1 & 1 & -1 & 1 \cr
1 & 1 & -1 & -1 \cr 0 & 3 & 1 & 0 \end{array}\right), \
 V_2: \left(\begin{array}{rrrr}
1 & 1 & 1 & 1 \cr 1 & 0 & 0 & -3 \cr
-1 & 2 & 0 & -1 \cr 1 & 1 & -1 & 1 \end{array}\right), \
 V_3: \left(\begin{array}{rr} 1 & 1 \cr 1 & -1 \end{array}\right).\]

We sum up the result: with
$z_{i,j}=p_iq_j-q_ip_j, \ 1\le i<j\le 5$
the following identity holds for the BW form generated by
two third order Hankel matrices:
\begin{eqnarray*}&&2z_{1,2}^2+3z_{1,3}^2+6z_{2,3}^2+2z_{1,4}^2+
4z_{2,4}^2+6z_{3,4}^2+z_{1,5}^2+2z_{2,5}^2+3z_{3,5}^2+2z_{4,5}^2\cr
&&-(z_{1,3}+z_{2,4}+z_{3,5})^2-(z_{1,2}+z_{2,3}+z_{3,4})^2-
(z_{2,3}+z_{3,4}+z_{4,5})^2=\cr \
&& \ \ \ (z_{1,2}-z_{2,3}-z_{3,4}+z_{4,5})^2+3(z_{2,3}-z_{3,4})^2\cr
&&+\frac{1}{6}(z_{1,3}+2z_{1,5}+z_{3,5})^2
+\frac{3}{2}(z_{1,3}-z_{3,5})^2+\frac{1}{3}
(z_{1,3}-3z_{2,4}-z_{1,5}+z_{3,5})^2\cr
&&+\frac{1}{2}(z_{1,4}+z_{2,5})^2+\frac{3}{2}(z_{1,4}-z_{2,5})^2.
\end{eqnarray*}  \end{example}

\section{Toeplitz matrices}\label{sec5}
In this section $P$ and $Q$ will be arbitrary real Toeplitz matrices.
Observe that the main diagonal entries don't play any role
(to prove this use temporarily the more standard notation
$P=(p_{j-i})$ and $Q=(q_{j-i}),$ then the $(i,j)$ entry in
$R=[P,Q]$ equals $\sum z_{k-i,j-k},$ while $z_{0,j-i}+z_{j-i,0}=0$).
Hence we can reduce the number of variables to get e.g. for $n=3$
the  index matrix
$IND=\Big(\begin{smallmatrix} 0 & 1 & 2 \cr 3 & 0 & 1
\cr 4 & 3 & 0 \cr \end{smallmatrix} \Big).$

Another speciality is that now there occur repeated terms as well.
To handle these, introduce the multiplicity vector $\mu$
of dimension $m$ by
\[\mu_i=\{\mathrm{the\ number\ of\ occurrences\ of\ } p_i
\mathrm{\ in\ } P, \ 1\le i\le m \}.\]
Then the following easily proved representation holds.

\begin{lemma}\[2 \, \Big(||P||^2||Q||^2- \mathrm{trace}^2(P^{T}Q)\Big) =
2 \sum_{i=1}^{m-1} \sum_{j=i+1}^m \mu_i \mu_j z_{i,j}^2, \]
and (since the commutator is skew persymmetric), \
$\|R\|^2 = 2 \sum_{i=1}^2\sum_{j=1}^{n-i} r_{i,j}^2. $
\end{lemma}

In view of the lemma, we define the symmetric matrix $C$
by help of equation $ z^T C z = \frac{1}{2}\,BW(p,q).$
Then there are $m=2(n-1)$ possible nonzero elements, the
candidate vector $z$ has dimension $N=\binom{m}{2},$
 and the total number of constraints $(A_t)$ is
 $M=\binom{m}{4}.$ For this problem we formulate a
 'quasi-optimal' strategy of choosing the dual variables.
 \medskip

{\it Strategy B.} \ Since the entries of $[P,Q]$ are linear forms
in $(z_{i,j}),$ their squares figuring in $\|R\|^2$ involve some
mixed products of the form $\pm\, 2\, z_{i,j}z_{k,l}. $
Whenever finding such a term with distinct \{$i,j,k,l$\}, we
increase the actual value of $y.$

It turns out that Strategy B works for $n, \ 3\le n\le 7,$
however, for $n\ge 8$ the dual matrix $S=C-\sum y_t A_t$
will have (more and more) negative eigenvalues.

\begin{lemma} For orders $n$ not exceeding 7, the matrix $S$
is p.s.d, i.e. for these values the BW form is sos.
Some further properties of $S$ of arbitrary order $n$ are:
the minimum off-diagonal entry of $S$ is
$-\lfloor\frac{n-1}{2}\rfloor;$ the defect of $S,$ i.e. the
multiplicity of zero as eigenvalue is $n-1;$ the maximal
diagonal entry and also the maximal eigenvalue is $n(n-2).$
Moreover, $S$ is a direct sum of two types
  of submatrices of the following order:

\noindent -- type (a): $2, 4, 6, \dots, 2(n-2);$
(denote by $B$ the largest block here)

\noindent -- type (b): $1, 1, 2, 2, \dots,n-2, n-2, n-1.$

\noindent The orders of these matrices sum up to
$(n-1)(2n-3),$ the order of $S.$
\end{lemma}

The largest block $B$ of type (a) is crucial.
It has a decomposition
$B=\Big( \begin{smallmatrix} D & H\cr H & D \end{smallmatrix}\Big ),$
with $D$ diagonal, $H$ Hermitian (i.e. symmetric), both of order $n-2.$
The diagonal elements of $D$ are $(i(i+1))$ in reverse order:
$((n-2)(n-1), \dots,6, 2).$ Matrix $H$ is also of a special structure:
the elements on the border are $-1,$ those on the 'neighboring'
border are $-2,$ and so on. This matrix is p.s.d. for $n\le 7,$
but has at least one negative eigenvalue for $n\ge 8.$

\begin{remark} The further submatrices of type (a) also are critical,
e.g. the next one (of order $2(n-3)$) has a similar form with
diagonal elements $(i(i+2))$ in $D,$ while $H$ is the same
(of the appropriate size). Therefore there is a second negative
eigenvalue for $n, \ 14\le n\le 20,$ and so on.
In general, the symmetric matrices $H$ are of the same form,
and the diagonal entries of $D$ are $(i(i+k))_i.$ \end{remark}

\begin{example} Matrices of order $5.$
In this case $P$ and  $Q$ have $m=2(n-1)=8$ nonzero elements,
the candidate vector $z$ has dimension $\binom{m}{2}=28,$
the total number of constraints is $\binom{m}{4}=70,$ and
the number of active constraints is $14.$

It always suffices to examine the first row and the
first column of $R,$ for all other entries are contained in
these, e.g. $R(1,1)=z_{1,5}+z_{2,6}+z_{3,7}+z_{4,8},$ and
$R(2,2)=z_{2,6}+z_{3,7}.$
The number of the active constraints for $n=5$, coming from
row 1 and column 1 is indeed 6 + 2 (3 + 1) = 14, as stated above.
This can be proved by induction, by noting that
\[\binom{n-1}{2} + 2\Big\{\binom{n-2}{2}+\binom{n-3}{2}+\dots
+\binom{2}{2}\Big\}= \frac{1}{6}(n-1)(n-2)(2n-3).\]

As regards the $y$ coordinates, since the product $2z_{2,6}z_{3,7}$
occurs two times (as the above formulae show), we write $-2$ in the
suitable positions (overwriting the $-1$-s) to get
$S(13,17) = S(3,21) = -2,$ and so on.  \end{example}

Now we give another example illustrating the role of the
crucial block $B.$

\begin{example} The case  $n=8.$
The matrices $D$ and $H$ are now:
\[ D=\left(\begin{array} {rrrrrr}
    42  &   0 &    0  &   0 &    0   &  0  \cr
     0  &  30 &    0  &   0 &    0   &  0  \cr
     0  &   0 &   20  &   0 &    0   &  0  \cr
     0  &   0 &    0  &  12 &    0   &  0  \cr
     0  &   0 &    0  &   0 &    6   &  0  \cr
     0  &   0 &    0  &   0 &    0   &  2
\end{array}\right), \
H=\begin{pmatrix}
    -1  &  -1 &   -1  &  -1  &  -1 &   -1 \cr
    -1  &  -2 &   -2  &  -2  &  -2 &   -1 \cr
    -1  &  -2 &   -3  &  -3  &  -2 &   -1 \cr
    -1  &  -2 &   -3  &  -3  &  -2 &   -1 \cr
    -1  &  -2 &   -2  &  -2  &  -2 &   -1 \cr
    -1  &  -1 &   -1  &  -1  &  -1 &   -1
\end{pmatrix}. \]
The characteristic polynomial of the block
$B=\Big( \begin{smallmatrix} D & H\cr H & D \end{smallmatrix}\Big )$
factorizes into $p_1p_2,$ where
$p_1(x)=x^6-100x^5+536x^4-53472x^3+327472x^2-575680x-145152,$
and $p_1$ has a negative zero: --\,0.2228.
(All other zeroes of $p_1$ and $p_2$ are positive.)
\end{example}

Finally we mention that although the above strategy works only
up to $n=7,$ the standard semidefinite program yields
results indicating that BW can be sos for some larger
orders, too, hence we guess that BW is sos in general.
The difficulty is that the corresponding dual variables
$(y_t)$ of the program are not recognizable real numbers.
Nevertheless we formulate the following.

\begin{conjecture} \label{conj} The B\"ottcher-Wenzel form (\ref{bwform})
generated by two  real Toeplitz matrices is sos, i.e.
a sum of squares of polynomials, now: quadratic forms.
Give -- if possible -- a rational certification, i.e. rational parameters
$(y_t)$ such that $S=C-\sum y_t A_t$ is positive semidefinite.
\end{conjecture}

\bigskip
\rightline{\emph{Received: February 8,  2012 {\tiny \raisebox{2pt}{$\bullet$\!}} Revised: March 22, 2012}} 

\end{document}